\documentclass[conference,letterpaper]{IEEEtran}

\usepackage[utf8]{inputenc} 
\usepackage[T1]{fontenc}
\usepackage{url}
\usepackage{ifthen}
\usepackage{cite}
\usepackage[cmex10]{amsmath} % Use the [cmex10] option to ensure complicance
\usepackage{amssymb}
                             
\usepackage{enumitem}                             
                             
\usepackage{microtype} % margin kerning
\usepackage{ragged2e} % better ragged paragraphs

\usepackage[pdftex]{graphicx}
\newtheorem{corollary}{Corollary}
\newtheorem{definition}{Definition}
\newtheorem{theorem}{Theorem}
\newtheorem{lemma}{Lemma}

\DeclareMathOperator{\Tr}{Tr}\let\tr\Tr

\def\mylangle{\raisebox{.5\depth}{\resizebox*{1.2\width}{.7\totalheight}{$\langle$}}}
\def\myrangle{\raisebox{.5\depth}{\resizebox*{1.2\width}{.7\totalheight}{$\rangle$}}}
\def\myvert{\raisebox{.5\depth}{\resizebox*{\width}{.7\totalheight}{$|$}}}
\def\myouter{\myrangle\mkern-5.65mu\mylangle}

\def\ket#1{{\myvert{#1}\myrangle}}

\def\braopket#1#2#3{\mathinner{\mylangle{#1}\myvert{#2}\myvert{#3}\myrangle}}
\def\ketbra#1#2{{\myvert{#1}\myouter{#2}\myvert}}

\def\tsum{\textstyle\sum\limits}

\begin{document}
\title{From Classical to Quantum:\\
Explicit Classical Distributions Achieving\\ Maximal Quantum $f$-Divergence} 

% %%% Single author, or several authors with same affiliation:
 \author{%
  \IEEEauthorblockN{Dimitri Lanier, Julien Béguinot, and Olivier Rioul}
 \IEEEauthorblockA{LTCI, Télécom Paris, Institut Polytechnique de Paris, France\\
                       Email: firstname.lastname@telecom-paris.fr}}

%%% Several authors with up to three affiliations:
%\author{%
%  \IEEEauthorblockN{Author 1}
%  \IEEEauthorblockA{Department of Electrical Engineering \\
%                    University 1\\
%                    City 1\\
%                    Email: author1@university1.edu}
%  \and
%  \IEEEauthorblockN{Author 2 and Author 3}
%  \IEEEauthorblockA{Research Center XY\\ 
%                    City 2\\
%                    Email: \{author2, author3\}@research-center.com}
%}

\maketitle

%%%%%%
%% Abstract: 
%% If your paper is eligible for the student paper award, please add
%% the comment "THIS PAPER IS ELIGIBLE FOR THE STUDENT PAPER
%% AWARD." as a first line in the abstract. 
%% For the final version of the accepted paper, please do not forget
%% to remove this comment!
%%

\begin{abstract}
THIS PAPER IS ELIGIBLE FOR THE STUDENT PAPER AWARD.
Explicit classical states achieving maximal $f$-divergence are given, allowing for a simple proof of Matsumoto's Theorem, and the systematic extension of any inequality between classical $f$-divergences 
to quantum $f$-divergences. Our methodology is particularly simple as it does not require any elaborate matrix analysis machinery but only basic linear algebra. It is also effective, as illustrated by two examples improving existing bounds:
(i)~an improved quantum Pinsker inequality is derived between $\chi^2$ and trace norm, and leveraged to improve a bound in decoherence theory;
(ii)~a new reverse quantum Pinsker inequality is derived for any quantum $f$-divergence, and compared to previous (Audenaert-Eisert and Hirche-Tomamichel) bounds. 
\end{abstract}

%The deadline for registering the paper and uploading the manuscript is \textbf{January 22, 2025} (anywhere on earth). \textbf{No extensions will be given.}
%A paper's primary content is restricted in length to \textbf{5 pages}, but authors are allowed an optional 6th page only containing references. The submission may contain a link to a longer, online version of the submitted manuscript, if the authors wish to include one.
% Only electronic submissions in form of a PDF file will be
%accepted. The PDF file has to be PDF/A compliant. A common problem is
%missing fonts. Make sure that all fonts are embedded. (In some cases,
%printing a PDF to a PostScript file, and then creating a new PDF with
%Acrobat Distiller, may do the trick.) More information (including
%suitable Acrobat Distiller Settings) is available from the IEEE
%website \cite{IEEE:pdfsettings, IEEE:AuthorToolbox}.

\section{Introduction}

%\subsection{}

Let $f(x)$ be a convex, real-valued function defined  for $x\geq0$ and satisfying $f(1)=0$.
We consider the following definition.
\begin{definition}[Classical $f$-Divergence~\cite{Csiszar1967}]
Let $p,q$ be classical distributions $p=(p_1,p_2,\ldots,p_n)$ and $q=(q_1,q_2,\ldots,q_n)$. The $f$-divergence between $p$ and $q$ is
\begin{equation}
D_f(p\|q) = \sum_i f\bigl(\frac{p_i}{q_i}\bigr) \, q_i 
\end{equation}
where by convention, summation is over all $i=1,2,\ldots,n$.
\end{definition}
As is well known, $f$-divergence is nonnegative and satisfies the data processing inequality 
\begin{equation}
D_f(W(p)\|W(q)) \leq D_f(p\|q)   
\end{equation}
for any classical channel $W$.
Classical $f$-divergences have been extensively studied in information theory, and many inequalities between various $f$-divergences have been established. 

In quantum information theory, classical probability distributions are replaced by (mixed) quantum states represented by density matrices, and several generalizations of $f$-divergences have been proposed \cite{optimized,PETZ198657,Hirche_2024,origin_max_div,studied_max_div,matsumoto2018}. 
We adopt the following general definition. 

\begin{definition}[Quantum $f$-Divergence]
Let $\rho,\sigma$ denote (mixed) quantum states in a Hilbert space of dimension $n$. 
A quantum $f$-divergence $D_f(\rho\|\sigma)$ is any nonnegative measure of $\rho,\sigma$ such that 
\begin{enumerate}[label=(\roman*)]
\item it coincides with the classical $f$-divergence $D_f(p\|q)$ when $\rho$ and $\sigma$ commute ($\rho\sigma=\sigma\rho$) so that they reduce to classical states with eigenvalues $p=\lambda(\rho)$ and $q=\lambda(\sigma)$; 
\item it satisfies the quantum data processing inequality 
\begin{equation}
D_f(\mathcal{W}(\rho)\|\mathcal{W}(\sigma)) \leq D_f(\rho\|\sigma)   
\end{equation}
for any quantum (completely positive and trace preserving) channel $\mathcal{W}$. 
\end{enumerate}
\end{definition}

The \emph{maximal $f$-divergence}, first proposed in \cite{origin_max_div}, further studied in \cite{studied_max_div} and popularized by Matsumoto \cite{matsumoto2018} stands out as a particularly important example of quantum $f$-divergence, due to its property of being the only quantum generalization that is maximal among all those satisfying condition (i) above.
% that reduce to the classical $f$-divergence when restricted to commuting density matrices.
%
The maximal $f$-divergence has many interesting properties, including monotonicity under quantum channels (ii) when $f$ is operator convex~\cite{origin_max_div}, coincidence with quantum $\chi^2$ divergence for $f(x)=x^2-1$, as well as coincidence with trace norm for $f(x)=|x-1|$ on a significant set of states~\cite{matsumoto2018}.

In the remainder of this paper, we revisit Matsumoto's maximal divergence theorem~\cite{matsumoto2018} by providing explicit \emph{classical} distributions achieving maximal $f$-divergence.
In particular, this gives a simple proof of Matsumoto's Theorem, including the quantum data processing inequality for maximal $f$-divergence, and the possibility of systematically extending any inequality between classical $f$-divergences to quantum $f$-divergences. Our methodology, detailed in Section~\ref{sec:main} is particularly simple, and  also effective as shown by two applications in Section~\ref{sec:application}.

\section{Explicit Maximal Divergence Theorem}\label{sec:main}

\subsection{Derivation}

The main result of the paper revisits Matsumoto's maximal divergence theorem~\cite{matsumoto2018} 
by providing explicit classical distributions achieving the maximal $f$-divergence, which allows a simple proof.

\begin{theorem}
\label{theorem:maintheo}
\begin{equation}\label{eq:maxfdiv} 
D_f^{\max} (\rho\|\sigma)  \triangleq \Tr\bigl( \sigma^{1/2} f(\sigma^{-1/2} \rho\sigma^{-1/2}) \sigma^{1/2}\bigr)
\end{equation}
is maximal over all quantum $f$-divergences $D_f(\rho\|\sigma)$:
\begin{equation}\label{ineq:maxdiv}
D_f(\rho\|\sigma)\leq D_f^{\max} (\rho\|\sigma)
\end{equation}
for any states $\rho,\sigma$ (where state $\sigma$ is invertible). 
%Furthermore, if $f$ is operator convex, it has been proven that $D_f^{\max}$ satisfies the data processing inequality \cite{origin_max_div}.
Furthermore, there exist classical distributions $r,s$ depending only on $\rho,\sigma$ (explicit expressions~\eqref{eq:expressions} given below in the proof) such that
\begin{equation}
D_f^{\max} (\rho\|\sigma)=D_f(r\|s).
\end{equation}
%for any $\rho,\sigma$. 
\end{theorem}
\begin{proof}
Let
\begin{equation}
\sigma^{-1/2} \rho\sigma^{-1/2} = \sum_i \lambda_i  \ketbra{u_i}{u_i} 
\end{equation}
be the spectral decomposition of  $\sigma^{-1/2} \rho\sigma^{-1/2}$ with eigenvalues $\lambda_i\geq 0$ and 
orthonormal eigenvectors $\ket{u_i}$, and define
\begin{equation}
\begin{aligned}
s_i &\triangleq \braopket{u_i}{\sigma}{u_i}  > 0\\
r_i &\triangleq \lambda_i s_i \geq 0.
\end{aligned}\label{eq:expressions}
\end{equation}
Observe that 
$r=(r_1,r_2,\ldots,r_n)$ and $s=(s_1,s_2,\ldots,s_n)$ are classical distributions, since 
\begin{equation}
\tsum_i s_i = \tsum_i \Tr(\sigma \ketbra{u_i}{u_i}) =  \Tr(\sigma \tsum_i\ketbra{u_i}{u_i})=\Tr(\sigma)=1
\end{equation}
and
\begin{equation}
\begin{aligned}
\tsum_i r_i &= \tsum_i \lambda_i \Tr(\sigma \ketbra{u_i}{u_i})= \Tr(\sigma \tsum_i\lambda_i\ketbra{u_i}{u_i})
\\
&=\Tr(\sigma \sigma^{-1/2} \rho\sigma^{-1/2}) = \Tr(\rho)=1.
\end{aligned}
\end{equation}
Also define
\begin{equation}
A_i = \frac{1}{\sqrt{s_i}}\sigma^{1/2}  \ketbra{u_i}{u_i}
\end{equation}
for $i=1,2,\ldots,n$.
The following operator sum representation
\begin{equation}\label{eq:channelV}
\mathcal{V}(\tau) = \sum_i A_i \tau A_i^\dag
\end{equation}
defines a quantum channel $\mathcal{V}$, which is trace preserving since
\begin{equation}
\tsum_i A_i^\dag A_i = \sum_i \frac{1}{s_i}\ketbra{u_i}{u_i} \sigma   \ketbra{u_i}{u_i} = \sum_i \ketbra{u_i}{u_i}=I.
\end{equation}
Now since the $\ket{u_i}$ are orthonormal, for any $k\ne i$, $A_k\ket{u_i}=0$. It follows that
\begin{equation}
\begin{aligned}
 \mathcal{V}(s)&=\mathcal{V}(\tsum_i s_i \ketbra{u_i}{u_i}) =  \tsum_i s_i A_i\ketbra{u_i}{u_i}A_i^\dag\\
 &= \tsum_i \sigma^{1/2}\ketbra{u_i}{u_i}\sigma^{1/2} = \sigma^{1/2}\sigma^{1/2}=\sigma 
\end{aligned}
\end{equation}
and similarly,
\begin{equation}
\begin{aligned}
 \mathcal{V}(r)&= \tsum_i r_i A_i\ketbra{u_i}{u_i}A_i^\dag = 
 \tsum_i \lambda_i\sigma^{1/2}\ketbra{u_i}{u_i}\sigma^{1/2} \\
 &= \sigma^{1/2}(\sigma^{-1/2} \rho\sigma^{-1/2})\sigma^{1/2}=\rho.
\end{aligned}
\end{equation}
Therefore, by the quantum data processing inequality~(ii) satisfied by $D_f(\rho\|\sigma)$, one has
\begin{equation}
D_f(\rho\|\sigma) = D_f(\mathcal{V}(r)\|\mathcal{V}(s))\leq D_f(r\|s) 
\end{equation}
where
\begin{equation}
\begin{aligned}
D_f(r\|s) &= \tsum_i f\bigl(\frac{r_i}{s_i}\bigr) s_i  \\&= \tsum_i f\bigl(\lambda_i\bigr) \Tr(\ketbra{u_i}{u_i}\sigma) \\
&= \Tr \bigl(\tsum_i f\bigl(\lambda_i\bigr) \ketbra{u_i}{u_i}\sigma\bigr)
\\
&=\Tr\bigl( f(\sigma^{-1/2} \rho\sigma^{-1/2}) \sigma\bigr),
\end{aligned}
\end{equation}
which proves~\eqref{ineq:maxdiv} and the Theorem.
\end{proof}

\subsection{Alternative Expressions and a Few Examples}

Notice that from Lemma~\ref{lem:identity} in the Appendix, 
\begin{equation}
\begin{aligned}
 D_f^{\max} (\rho\|\sigma)&= \Tr\bigl( \sigma^{1/2} f(\sigma^{-1/2} \rho\sigma^{-1/2}) \sigma^{1/2}\bigr)\\
 &=\Tr\bigl( \sigma f(\sigma^{-1} \rho) \bigr)\\
 &=\Tr\bigl( f( \rho\sigma^{-1}) \sigma\bigr).
\end{aligned}
\end{equation}
Additionally if $f(x)=x\, g(x)$, using~\eqref{eq:identity2},
\begin{equation}
\begin{aligned}
 D_f^{\max} (\rho\|\sigma)&= \tr \bigl(\rho g(\sigma^{-1} \rho)\bigr)
\\&= \tr \bigl(\rho g(\rho\sigma^{-1})\bigr)
\\&= \tr \bigl(\rho^{1/2} g(\rho^{1/2}\sigma^{-1}\rho^{1/2}) \rho^{1/2}\bigr).
\end{aligned}
\end{equation}
The symmetrical expressions such as~\eqref{eq:maxfdiv} are often preferred because the are explicitly written in terms of positive operators $\sigma^{-1/2} \rho\sigma^{-1/2}\geq 0$, 
$\rho^{1/2}\sigma^{-1}\rho^{1/2}\geq 0$, etc. We give three examples.

\subsubsection{$f(x)=x\log(x)$}
The well-known quantum extension of the Kullback-Leibler divergence $D(p\|q)=\sum_i \log(\frac{p_i}{q_i})\,p_i$ is the quantum relative entropy
\begin{equation}
D(\rho\|\sigma) \triangleq  \Tr\bigl(\rho(\log\rho-\log\sigma)\bigr)
\end{equation}
also noted $S(\rho\|\sigma)$. The corresponding maximal $f$-divergence is $D^{\max}(\rho\|\sigma)= \tr \bigl(\rho \log(\rho\sigma^{-1})\bigr)$, or in symmetric form
\begin{equation}
D^{\max}(\rho\|\sigma)=\Tr\bigl(\rho^{1/2}\log(\rho^{1/2}\sigma^{-1}\rho^{1/2})\rho^{1/2}\bigr).
\end{equation}
It was studied in \cite{belavkin2002entangled} and plays a role in quantum statistical estimation \cite{holevo2011probabilistic}. 

This maximal divergence  is not to be confused with the quantum \emph{max-relative entropy} or \emph{max-divergence}~\cite{max_rel,Datta2009}:
\begin{equation}\label{eq:maxrelent}
    D_{\max}(\rho \| \sigma) \triangleq \log \lambda_{\max} (\rho \sigma^{-1})
    =  \log \lambda_{\max} (\sigma^{-\frac{1}{2}} \rho \sigma^{-\frac{1}{2}}) 
\end{equation}
(where $\lambda_{\max}$ denotes the largest eigenvalue),
which is not a maximal $f$-divergence but rather corresponds to the Rényi divergence of order $\infty$.

\subsubsection{$f(x)=x^2-1$ (or $(x-1)^2$)} The chi-squared divergence 
\begin{equation}\label{eq:chi2}
\chi^2(p || q) = \sum_i \frac{(p_i-q_i)^2}{q_i} = \sum_i \frac{p_i^2}{q_i} - 1
\end{equation}
has the natural quantum extension
\begin{equation}\label{eq:chi2Q}
    \chi^2(\rho\|\sigma) =\Tr \bigl( (\rho-\sigma)^2\sigma^{-1}  \bigr) 
    =\Tr \bigl( \rho^2\sigma^{-1} \bigr) - 1.   
\end{equation}
It turns out that it also equals the corresponding maximal $f$-divergence, since
\begin{equation}\label{eq:chi2max}
\begin{aligned}
    (\chi^2)^{\max}(\rho || \sigma) &=\Tr \bigl( (\rho\sigma^{-1})^2\sigma - \sigma  \bigr)\\
    &= \Tr \bigl( \rho\sigma^{-1}\rho \bigr)-\Tr(\sigma) \\
    &= \Tr\bigl(\rho^2\sigma^{-1}\bigr)-1\\
    &= \chi^2(\rho\|\sigma) .
\end{aligned}
\end{equation}

\subsubsection{$f(x)=|x-1|$} The natural quantum extension of the total variation distance $\|p-q\|_1 = \sum_i |p_i-q_i|$ is the trace distance
\begin{equation}
\| \rho-\sigma\|_1 = \tr (|\rho-\sigma|),
\end{equation}
while the corresponding maximal $f$-divergence is
\begin{equation}
\begin{aligned}
 \|\rho-\sigma\|_1^{\max}&=\tr\bigl( |\rho\sigma^{-1} - 1| \sigma \bigr)
 \\&= \tr\bigl(\sigma^{1/2}|\sigma^{-1/2}(\rho-\sigma)\sigma^{-1/2}|\sigma^{1/2}\bigr).
\end{aligned}
\end{equation}

\subsection{Data Processing Inequality for Maximal $f$-Divergence}

In this Subsection we assume that $f$ is \emph{operator convex}, that is, for any positive operators $X_i\geq 0$, and real numbers $\lambda_i \geq 0$ such that $\sum_i\lambda_i=1$, one has
\begin{equation}
f\bigl(\smash[b]{\sum_i} \lambda_i X_i \bigr) \leq \smash[b]{\sum_i} \lambda_i f(X_i)
\end{equation}
where $\leq$ denotes the Loewner order. By Lemma~\ref{lem:opconvex} we also have a variant of Jensen's inequality the form
\begin{equation}
f\bigl(\smash[b]{\sum_i}\Lambda_i x_i \bigr) \leq \smash[b]{\sum_i} \Lambda_i f(x_i)
\end{equation}
for any real numbers $x_i\geq 0$ and positive operators $\Lambda_i\geq 0$ such that $\sum_i \Lambda_i=I$.

We have the following extension of Theorem~\ref{theorem:maintheo}:
\begin{theorem}[Data Processing Inequalities]\label{thm:dpi}
Let $\rho,\sigma$, $r,s$ be as in Theorem~\ref{theorem:maintheo} and consider any quantum channel $\mathcal{W}$. Then
\begin{equation}\label{eq:dpirs}
D_f^{\max}(\mathcal{W}(r)\|\mathcal{W}(s)) \leq  D_f^{\max}(r\|s).
\end{equation}
In other words, letting $\tilde{\rho}=\mathcal{W}(r)$ and $\tilde{\sigma}=\mathcal{W}(s)$ we have
\begin{equation}\label{eq:dpirs'}
D_f^{\max}(\tilde{\rho}\|\tilde{\sigma}) \leq  D_f^{\max}(\rho\|\sigma).
\end{equation}
\end{theorem}
By Theorem~\ref{theorem:maintheo}, equality holds in~\eqref{eq:dpirs} or~\eqref{eq:dpirs'} if $\mathcal{W}=\mathcal{V}$ is defined by~\eqref{eq:channelV}.

In addition, by composition of channels $\mathcal{W}\circ\mathcal{V}$, we may consider $\tilde{\rho}=\mathcal{W}(\rho)=\mathcal{W}\bigl(\mathcal{V}(r)\bigr)$ and $\tilde{\sigma}=\mathcal{W}(\sigma)=\mathcal{W}\bigl(\mathcal{V}(s)\bigr)$. Then \eqref{eq:dpirs'} gives the \emph{quantum data processing inequality}~(ii) for maximal $f$-divergence for any quantum channel $\mathcal{W}$:
\begin{equation}
 D_f^{\max}(\mathcal{W}(\rho)\|\mathcal{W}(\sigma)) \leq  D_f^{\max}(\rho\|\sigma)
\end{equation}
which generalizes~\eqref{eq:dpirs} for any mixed states $\rho,\sigma$.

\begin{proof}[Proof of Theorem~\ref{thm:dpi}]
Proceed to prove  $D_f^{\max}(\tilde{\rho}\|\tilde{\sigma})\leq D_f^{\max}(r\|s)$ where $\tilde{\sigma}=\mathcal{W}(s)$ and
\begin{equation}
\tilde{\rho}= \mathcal{W}(r) = \mathcal{W}(\tsum_i r_i \ketbra{u_i}{u_i})=\tsum_i r_i \mathcal{W}(\ketbra{u_i}{u_i}).
\end{equation}
We have
\begin{equation}
\tilde{\sigma}^{-1/2} \tilde{\rho}\, \tilde{\sigma}^{-1/2} = \smash[b]{\sum_i}  \frac{r_i}{s_i} \cdot  \tau_i
\end{equation}
where $\tau_i=s_i \cdot \tilde{\sigma}^{-1/2} \mathcal{W}(\ketbra{u_i}{u_i}) \tilde{\sigma}^{-1/2}\geq 0$ are positive operators  such that 
\begin{equation}
\sum_i \tau_i  =  \tilde{\sigma}^{-1/2} \mathcal{W}(s) \tilde{\sigma}^{-1/2} = \tilde{\sigma}^{-1/2} \tilde{\sigma} \tilde{\sigma}^{-1/2} =I.
 \end{equation}
Therefore, by Lemma~\ref{lem:opconvex},
\begin{equation}
f(\tilde{\sigma}^{-1/2} \tilde{\rho}\, \tilde{\sigma}^{-1/2} ) \leq 
 \sum_i f(\frac{r_i}{s_i})  \tau_i.
\end{equation}
Multiplying by $\tilde{\sigma}$ and taking the trace, we obtain
\begin{equation}
 D_f^{\max}(\tilde{\rho}\|\tilde{\sigma})\leq \sum_i f(\frac{r_i}{s_i})  \tr(\tilde{\sigma}^{1/2} \tau_i \tilde{\sigma}^{1/2})
\end{equation}
where 
\begin{equation}
\tr(\tilde{\sigma}^{1/2} \tau_i \tilde{\sigma}^{1/2})=s_i\cdot \tr\bigl( \mathcal{W}(\ketbra{u_i}{u_i})\bigr)=s_i\cdot \tr(\ketbra{u_i}{u_i})=s_i. 
\end{equation}
This proves~\eqref{eq:dpirs}.
\end{proof}

%%% remarque OR (pas besoin de le mettre)
%\begin{remark}
%Applying the Corollary to any $p,q$ (the proof does not depend on the definition of $r$ and $s$), we obtain that $\mathcal{V}$ is the optimal reverse channel in the sense of Matsumoto. 
%\end{remark}

\subsection{Application: From Classical to Quantum}

The main application studied in this paper of our derivation of Theorem~\ref{theorem:maintheo} is to extend classical inequalities to the quantum setting, thanks to the following
\begin{corollary}[Classical to Quantum]
    \label{ctq}
If for some function $\phi$ the inequality between classical $f$-divergences 
\begin{equation}
D_f(p\|q)\leq \phi(D_g(p\|q))
\end{equation}
holds for any classical distributions $p,q$,
then any quantum $f$-divergence $D_f(\rho\|\sigma)$ satisfies the inequality
\begin{equation}
   D_f(\rho\|\sigma)\leq\phi(D_g^{\max}(\rho\|\sigma))
\end{equation}
for any $\rho,~\sigma$ with \hbox{$D_g^{\max} (\rho\|\sigma) \!=\! \Tr\bigl( \sigma^{1/2} g(\sigma^{-1/2} \rho\sigma^{-1/2}) \sigma^{1/2}\bigr)$}.
\end{corollary}
\begin{proof}
By Theorem \ref{theorem:maintheo}, there exist classical distributions $r,s$ depending only on $\rho,\sigma$ such that
$D_f^{\max}(\rho\|\sigma)  = D_f(r\|s)$ and
$D_g^{\max}(\rho\|\sigma)  = D_g(r\|s)$.
Since by hypothesis we have 
$D_f(r\|s)\leq \phi(D_g(r\|s))$
it follows that
$
D_f(\rho\|\sigma)\leq D_f^{\max}(\rho\|\sigma)\leq\phi(D_g^{\max}(\rho\|\sigma))
$.
\end{proof}
In the sequel, we use Corollary~\ref{ctq} to generalize both classical Pinsker and reverse Pinsker inequalities to the quantum case.

\section{Applications to (Reverse) Pinkser Inequalities}
\label{sec:application}

\subsection{Improved Quantum Pinsker Inequality for $\chi^2$-divergence}
\label{sub-sec:Q-Pinsker}
Pinsker's inequality~\cite{Pinsker1960}, first established by Schützenberger~\cite{Schutzenberger54}---see \cite{Rioul2024} for an extensive historical review---is a fundamental result in information theory that provides a lower bound on the Kullback-Leibler divergence  between two probability distributions in terms of their total variation distance (or trace distance in the quantum case):
\begin{equation}
\begin{aligned}
D(p\|q) &\ge \frac{1}{2} \|p-q\|_1^2\\
D(\rho\|\sigma) &\ge \frac{1}{2} \|\rho-\sigma\|_1^2\\ 
\end{aligned}
\end{equation}
where divergences are expressed in nats (with logs to base $e$).

Pinsker's inequality  has been generalized to $f$-divergences, particularly to the chi-squared divergence 
defined by~\eqref{eq:chi2},~\eqref{eq:chi2Q}.
%for $f(x) = x^2 - 1$ defined by  %In the quantum setting the $\chi^2$ divergence is defined in \cite[Eq.~7 (where $\alpha=0$)]{quant-chi2} 
%\begin{equation}
%\begin{aligned}
%\chi^2(p || q) &= \sum_i \frac{(p_i-q_i)^2}{q_i} = \sum_i \frac{p_i^2}{q_i} - 1\\
%    \chi^2(\rho\|\sigma) 
%    &=\Tr \bigl( (\rho-\sigma)^2\sigma^{-1}  \bigr) 
%    =\Tr \bigl( \rho^2\sigma^{-1} \bigr) - 1.  
%\end{aligned}
%\end{equation}
A well-known Pinsker inequality for $\chi^2$-divergence is~\cite{quant-chi2}
\begin{equation}
\begin{aligned}
\chi^2(p || q) &\geq \|p - q\|_1^2\\
\chi^2(\rho || \sigma) &\geq \|\rho - \sigma\|_1^2. 
\end{aligned}
\end{equation}
Based on the following classical improvement of Pinsker's inequality~\cite[Eq. (138)]{classical_chi}.  
\begin{equation}\label{classical_chi_in}
\chi^2(p || q) \geq 
\begin{cases}
\|p - q\|_1^2, & \text{if } \|p - q\|_1 \in [0, 1], \\
\frac{\|p - q\|_1}{2 - \|p - q\|_1}, & \text{if } \|p - q\|_1 \in (1, 2].
\end{cases}
\end{equation}
we easily obtain the following
\begin{theorem}[Improved Quantum Pinsker Inequality for $\chi^2$-Divergence]\label{thm:improved-pinkser}
The quantum chi-squared divergence $\chi^2(\rho || \sigma)$ satisfies the following inequality
\begin{equation}
\chi^2(\rho || \sigma) \geq 
\begin{cases}
\|\rho - \sigma\|_1^2, & \text{if } \|\rho - \sigma\|_1 \in [0, 1], \\
\frac{\|\rho - \sigma\|_1}{2 - \|\rho - \sigma\|_1}, & \text{if } \|\rho - \sigma\|_1 \in (1, 2].
\end{cases}
\end{equation}
for any quantum states $\rho$, $\sigma$.
\end{theorem}

\begin{proof}
Observe that the quantum $\chi^2$-divergence is the maximal $f$-divergence for $f(x)=x^2-1$,
by~\eqref{eq:chi2max}.
%since~\eqref{eq:maxfdiv}  writes
%\begin{equation}
%\begin{aligned}
%    D_{x^2-1}^{\max}(\rho || \sigma) &=\Tr \bigl(\sigma((\sigma^{-1/2}\rho\sigma^{-1/2})^2-I)\bigr)\\
%    &= \Tr \bigl(\sigma  \sigma^{-1/2}\rho\sigma^{-1}\rho\sigma^{-1/2}\bigr)-\Tr(\sigma) \\
%    &= \Tr\bigl(\rho^2\sigma^{-1}\bigr)-1. 
%\end{aligned}
%\end{equation}
%Also note that the trace norm is a particular $f$-divergence with $f(x)=|x-1|$.
The result is then immediate from~\eqref{classical_chi_in} by Corollary~\ref{ctq} applied to $f(x)= |x-1|$ and $g(x)=x^2-1$.
\end{proof}

\paragraph*{Application}

Using Theorem~\ref{thm:improved-pinkser} we can improve Temme \textit{et al.} bound in decoherence theory~\cite{quant-chi2}. The authors consider a quantum system time evolution given by a Liouvillian of unique equilibrium state $\sigma$.
Letting $\rho_t$ the state of the quantum system at time $t$, the trace distance $\|\rho_t-\sigma\|_1$ eventually vanishes as time goes by.
It is important when designing quantum technology to evaluate the rate of convergence to $0$ of this quantity is order to estimate how long can one keep significant quantum information.
In \cite[Eq.~(43) for $k=0$]{quant-chi2} the authors show that
\begin{equation}
    \|\rho_t-\sigma \|_1^2 \leq \chi^2(\rho_t\|\sigma) \leq e^{-\lambda t} \chi^2(\rho_0\|\sigma)
\end{equation}
where $\lambda > 0$ is the second largest eigenvalue of the Liouvillian defined in their article. 
Now, improving the first inequality  $\|\rho-\sigma\|_1^2 \leq  \chi^2(\rho\|\sigma)$ using Theorem~\ref{thm:improved-pinkser}, one readily obtains the following improved bound for short times:
\begin{equation}
\|\rho_t - \sigma\|_1 \leq
\begin{cases}
\frac{2e^{-\lambda t} \chi^2(\rho_0 || \sigma)}{1 + e^{-\lambda t} \chi^2(\rho_0 || \sigma)}, & \text{if } e^{\lambda t}< \chi^2(\rho_0 || \sigma) \\
e^{-\frac{\lambda}{2} t} \chi(\rho_0 || \sigma), & \text{otherwise,} 
\end{cases}
\end{equation}
where we have noted $\chi=\sqrt{\chi^2}$.

Figure~\ref{fig:bound_comparison} illustrates this improvement, which is naturally all the more important as $\chi^2(\rho_0 || \sigma)$ is large, i.e.,   the starting state is far away from the stationary one.
\begin{figure}[h]
    \centering
    \includegraphics[width=0.49\textwidth,height=6cm]{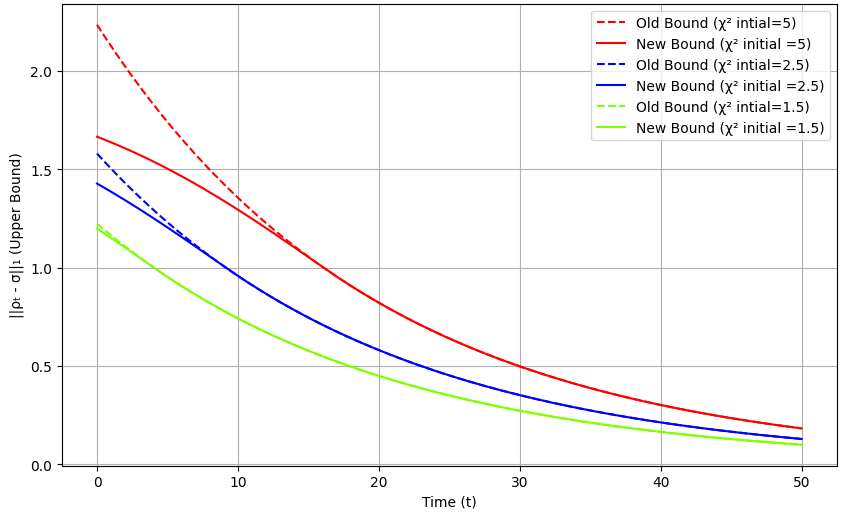}
    \caption{Comparison of the previous Temme \textit{et al.} bound  and the improved bound of this paper on the trace norm distance $||\rho_t - \sigma||_1$ as a function of time, when $\lambda = 0.1$. Different curves correspond to different initial values of the chi-squared divergence $\chi^2(\rho_0 || \sigma)$. The improved bound provides a tighter estimate at earlier times, particularly for larger initial $\chi^2(\rho_0 || \sigma)$.}
    \label{fig:bound_comparison}
\end{figure}

\subsection{Reverse Pinsker Inequalities for $f$-Divergences}

While Pinsker's inequality gives a lower bound on the $f$-divergence in terms of the total variation distance,
a \emph{reverse Pinsker inequality} gives an upper bound on the $f$-divergence, under some conditions on the distributions. This again can be extended in the quantum setting in the form of a upper bound in terms of the trace distance between states:
\begin{equation}
 D_f(\rho\|\sigma) \leq \phi (\|\rho - \sigma\|_1).
\end{equation}
%where in general, $\phi$ depends on $\rho$ and $\sigma$.
In this section, we use Corollary~\ref{ctq} to derive a reverse Pinsker inequality for any $f$-divergence, based on Binette's sharp (classical) reverse Pinsker inequality~\cite{Binette_2019}:
\begin{equation}\label{eq:Binette}
D_f(p\|q)
\leq
\frac{\|p-q\|_1}{2} \Bigl ( \frac{f(m)}{1-m} +\frac{f(M)}{M-1} \Bigr)
\end{equation}
where $m=\min_i\frac{p_i}{q_i}$ and $M=\max_i\frac{p_i}{q_i}$.
%we easily obtain the following 
\begin{theorem} \label{thm:QrevPinskerbound}
For any mixed states $\rho, \sigma$ satisfying the condition
\begin{equation}\label{eq:cns}
 | \rho-\sigma | \leq \rho+\sigma,
\end{equation}
we have
\begin{equation}\label{eq:QuantumBinette}
    D_f(\rho\|\sigma)
    \leq
     \frac{\|\rho-\sigma\|_1}{2}
    \biggl (\frac{f(m)}{1-m} +\frac{f(M)}{M-1} \biggr ).
\end{equation}
where $m=\lambda_{\min}(\sigma^{-\frac{1}{2}}\rho\sigma^{-\frac{1}{2}})$ 
and $ M = \lambda_{\max}(\sigma^{-\frac{1}{2}}\rho\sigma^{-\frac{1}{2}})$ are the smallest and largest eigenvalues of $\sigma^{-\frac{1}{2}}\rho\sigma^{-\frac{1}{2}}$, respectively. 
Furthermore, the inequality is sharp: equality holds for some $\rho,\sigma$ satisfying~\eqref{eq:cns}.
\end{theorem}

\begin{proof}
According to~\cite[Eq.~(64)]{matsumoto2018}, \eqref{eq:cns} is a (necessary and) sufficient condition for the maximal $g(x)=|x-1|$-divergence to coincide with the trace distance:
$D_g^{\max} (\rho\|\sigma) = \|\rho-\sigma\|_1$.
%A sufficient (but not necessary) condition~\cite[Eq.~(67)]{matsumoto2018} is $(\rho-\sigma)^2 \leq (\rho+\sigma)^2$, that is, $\rho\sigma+\sigma\rho\geq 0$.
%
Also observe that from the classical distribution representation of maximal $f$-divergence of Theorem~\ref{theorem:maintheo}, we have
\begin{equation}\label{eq:m-and-M}
\begin{aligned}
 m&\triangleq \min_i\frac{r_i}{s_i} = \min_i\lambda_i = \lambda_{\min}(\sigma^{-\frac{1}{2}}\rho\sigma^{-\frac{1}{2}})\\
 M&\triangleq \max_i\frac{r_i}{s_i} = \max_i\lambda_i = \lambda_{\max}(\sigma^{-\frac{1}{2}}\rho\sigma^{-\frac{1}{2}})
\end{aligned}
\end{equation}
which depend only on $\rho,\sigma$.
Therefore, we may apply Corollary~\ref{ctq} to $\phi (x)=\frac{x}{2} (\frac{f(m)}{1-m} +\frac{f(M)}{M-1} )$, $D_g^{\max} (\rho\|\sigma) = \|\rho-\sigma\|_1$.
This gives~\eqref{eq:QuantumBinette}. Equality is achieved for commuting states $\rho, \sigma$, which necessarily satisfy~\eqref{eq:cns} and correspond to the ternary classical distributions given in~\cite[Eq.~(16)]{Binette_2019} achieving equality in~\eqref{eq:Binette}.
\end{proof}

%In the particular case $f(x)=x\log x$ of von Neumann's relative entropy $D(\rho \| \sigma) =  \Tr\bigl(\rho(\log\rho-\log\sigma)\bigr)$, one obtains
%\begin{equation}\label{eq:QuantumBinetteVN}
%    D(\rho\|\sigma)
%    \leq
%     \frac{\|\rho-\sigma\|_1}{2}
%    \biggl (\frac{m\log m}{1-m} +\frac{M\log M}{M-1} \biggr )
%\end{equation}
%where the inequality may be strict for non commuting states $\rho, \sigma$ such that $D(\rho\|\sigma)<D_f^{\max}(\rho\|\sigma)=\Tr\bigl(\rho\log(\rho^{1/2}\sigma^{-1}\rho^{1/2})\bigr)$.
%The latter maximal divergence was studied in \cite{belavkin2002entangled} and plays a role in quantum statistical estimation \cite{holevo2011probabilistic}.

\subsection{Comparison to Hirche-Tomamichel's Bound}

First notice comparing~\eqref{eq:m-and-M} to~\eqref{eq:maxrelent} that in terms of the max-relative entropy,
\begin{equation}\label{eq:identificationDmax}
\begin{aligned}
M&=\exp \bigl( D_{\max}(\rho\|\sigma) \bigr)\\
m&=\exp \bigl(-D_{\max}(\sigma\|\rho) \bigr) 
\end{aligned}
\end{equation}
where the latter equality comes from 
$m=\lambda_{\min} \bigl ( \rho^{\frac{1}{2}}\sigma^{-1}\rho^{\frac{1}{2}} \bigr )=
\bigl(\lambda_{\max}(\rho^{-\frac{1}{2}}\sigma\rho^{-\frac{1}{2}})\bigr)^{-1}$.
Hirche and Tomamichel's bound~\cite[Prop.~5.2]{Hirche_2024} is
\begin{equation}
    D_f(\rho \| \sigma) 
    \leq 
    \frac{\|\rho - \sigma \|_1}{2}
    \zeta_1(\rho,\sigma)
\end{equation}
where $\zeta_1(\rho,\sigma)$ is defined by a complicated integral relation~\cite[Eq.~(5.12)]{Hirche_2024}
which from~\eqref{eq:identificationDmax} reduces to
\begin{align}
    \zeta_1(\rho,\sigma)\!&\triangleq\!\int_1^{M}\!\tfrac{M-\gamma}{M - 1} f''(\gamma) d\gamma 
    \!+\!
    \int_{1}^{\frac{1}{m}}
    \tfrac{\frac{1}{m} - \gamma}{\frac{1}{m} - 1}
    \gamma^{-3} f''(\gamma^{-1}) d\gamma
    \notag\\\label{eq:chgvar}
    &=
    \int_1^{M}\!\tfrac{M-\gamma}{M - 1} f''(\gamma) d\gamma 
    +
    \int_m^{1}\!\tfrac{\gamma - m}{1 - m} f''(\gamma) d\gamma 
    \\\label{eq:ipp}
    &=
    \frac{f(M)}{M-1} - f'(1) + \frac{f(m)}{1-m} + f'(1)
\end{align}
where we made the change of variable $\gamma \to \gamma^{-1}$ in the second integral in~\eqref{eq:chgvar},
and integrated by parts both integrals in~\eqref{eq:ipp}.
Thus as it turns out, we find that Hirche-Tomamichel's bound \emph{matches exactly} the bound of Theorem~\ref{thm:QrevPinskerbound}, which is the natural extension of Binette's bound to the quantum case.

There are some differences, however:
\begin{itemize}
    \item Hirche-Tomamichel's bound was derived only  for their specific $f$-divergence defined as an integral  over hockey stick divergences. Our bound holds in general for any $f$-divergence, in particular for the maximal $f$-divergence.
    \item Hirche and Tomamichel assume that $f$ is twice differentiable, which is not needed in this paper.
    \item on the downside, we assume condition~\eqref{eq:cns}, which is not needed by Hirche and Tomamichel. Our simulations show, however, that randomly drawn states in dimension~4 satisfy~\eqref{eq:cns} more than 80\% of the time. We conjecture that~\eqref{eq:cns} is in fact not necessary for~\eqref{eq:QuantumBinette} to hold. 
\end{itemize}

%Specializing $f(x) = x \log x$  in Eqn.~\eqref{eq:simp-binette} we indeed recover \cite[Eq. (1.12)]{Hirche_2024}.
%Using there notation we have : 
%\begin{equation}\label{eq:HTbound} 
%\begin{aligned}
%     D(\rho\|\sigma) 
%   \leq &\tfrac{\|\rho - \sigma\|_1}{2}\bigl( \tfrac{e^{D_{\max}(\rho\|\sigma)} D_{\max}(\rho\|\sigma)}{e^{D_{\max}(\rho\|\sigma)} - 1} + \tfrac{D_{\max}(\sigma\|\rho) e^{-D_{\max}(\sigma\|\rho)}}{e^{-D_{\max}(\sigma\|\rho)}-1} \bigr)
%    =   \\ 
%   &\tfrac{\|\rho - \sigma\|_1}{2}\bigl( \tfrac{e^{D_{\max}(\rho\|\sigma)} D_{\max}(\rho\|\sigma)}{e^{D_{\max}(\rho\|\sigma)} - 1} + \tfrac{D_{\max}(\sigma\|\rho)}{1 - e^{D_{\max}(\sigma\|\rho)}} \bigr)    
%\end{aligned}
%\end{equation}

\subsection{Comparison to Audenaert-Eisert's Bound}

In the particular case $f(x)=x\log x$, Audenaert \& Eisert~\cite[Thm.~1]{Audenaert_2011} derived the following bound:
\begin{equation}\label{eq:AEbound}
   D(\rho\|\sigma)
    \leq 
    (\beta+  \tfrac{\|\rho-\sigma\|_1}{2})\log(1+ \tfrac{\|\rho-\sigma\|_1}{2\beta}) 
    -  \alpha \log(1+ \tfrac{\|\rho-\sigma\|_1}{2\alpha})
\end{equation} 
where $\alpha=\lambda_{\min}(\rho)$ and $\beta=\lambda_{\min}(\sigma)$ and the second term vanishes for $\alpha=0$.
Fig.~\ref{fig:aud_comp} compares this bound to our bound~\eqref{eq:QuantumBinette} for $f(x)=x\log x$ on 10,000 randomly drawn $4\times4$ positive matrices $(\rho,\sigma)$ of unit trace satisfying~\eqref{eq:cns}.

Our Python code for all figures in this paper  is available here: \hbox{\url{https://anonymous.4open.science/r/code_rech-7004}}.

\begin{figure}[h!]
    \centering
    \includegraphics[scale=0.6]{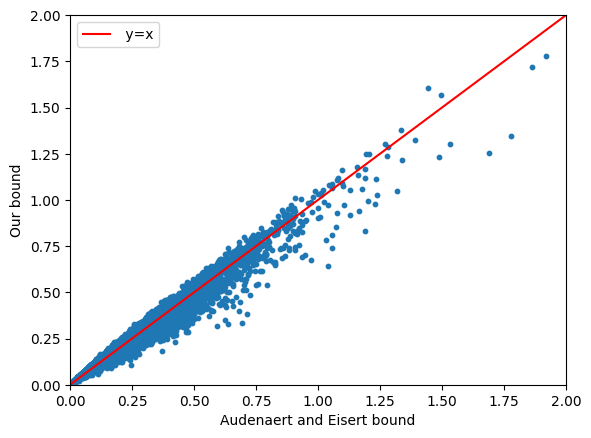}
    \caption{Scatter plot comparing our bound~\eqref{eq:QuantumBinette} to Audenaert-Eisert's bound~\eqref{eq:AEbound} on 10,000 random positive $4\times4$ matrices of unit trace. Points below the diagonal red line are instances where our bound~\eqref{eq:QuantumBinette} is tighter.}
    \label{fig:aud_comp}
\end{figure}
%
%%\item
%Hirche and Tomamichel~\cite[Eq. (1.12)]{Hirche_2024}:
%\begin{equation}\label{eq:HTbound}
%\begin{aligned}
 %  D&(\rho\|\sigma) 
   %\leq \tfrac{\|\rho - \sigma\|_1}{2}\bigl( \tfrac{e^{D_{\max}(\rho\|\sigma)} D_{\max}(\rho\|\sigma)}{e^{D_{\max}(\rho\|\sigma)} - 1} + \tfrac{D_{\max}(\sigma\|\rho)}{1 - e^{D_{\max}(\sigma\|\rho)}} \bigr) 
  % \\
  % &= \tfrac{\|\rho - \sigma\|_1}{2}\bigl( \tfrac{e^{D_{\max}(\rho\|\sigma)} D_{\max}(\rho\|\sigma)}{e^{D_{\max}(\rho\|\sigma)} - 1} + \tfrac{- D_{\max}(\sigma\|\rho) e^{-D_{\max}(\sigma\|\rho)}}{1-e^{-D_{\max}(\sigma\|\rho)}} \bigr) 
%\end{aligned}
%\end{equation}
%where the so-called max-relative entropy~\cite{max_rel,Datta2009} (not a maximal divergence) is 
 %$D_{\max}(\rho\|\sigma)= \log \bigl(\min\{\lambda \mid \rho \leq \lambda \sigma\}\bigr)=\log\max\lambda(\sigma^{-\frac{1}{2}}\rho\sigma^{-\frac{1}{2}})$.
%The result of the comparison is illustrated in a scatter plot in Fig.~\ref{fig:Hirche_comp}. 
%
%\end{enumerate}

\section{Conclusion}

In this work, we have revisited Matsumoto's maximal divergence theorem by providing explicit classical distributions achieving maximal $f$-divergence. 
Leveraging these distributions, we showed that classical $f$-divergence inequalities 
can be easily extended to their quantum counterparts. This approach simplifies the derivation of quantum inequalities, as it avoids complex matrix analysis and relies instead on simple linear algebra.
We applied our methodology to find both quantum Pinsker and quantum reverse Pinsker inequalities.
In particular, the reverse Pinsker inequality of Hirche and Tomamichel was found to match the quantum generalization of the classical reverse Pinsker inequality of Binette.
Our quantum Pinsker inequality for $\chi^2$ divergence was also applied to improve a bound on the evolution of quantum states under dynamical equations, which illustrates the operational relevance of our approach.

%%%%%%
%% Appendix:
%% If needed a single appendix is created by
%%
\appendix
%%
%% If several appendices are needed, then the command
%%
% \appendices
%%
%% in combination with further \section commands can be used.
%%%%%%

The following Lemmas seem well known in certain circles, yet we could not find a precise reference.
\begin{lemma}\label{lem:identity}
Let $A,B$ be normal operators in a Hilbert space of dimension $n$. Then for any functions $f,g,h:\mathbb{C}\to\mathbb{C}$,
\begin{equation}
\tr \bigl(g(A)f(AB)h(A)\bigr) = \tr \bigl(g(A)f(BA)h(A)\bigr)
\end{equation}
\end{lemma}
\begin{proof}
Let  $AB= \sum_i \lambda_i  \ketbra{u_i}{u_i}$
be the spectral decomposition of  $AB$ with eigenvalues $\lambda_i\geq 0$ and 
orthonormal eigenvectors $\ket{u_i}$, so that $f(AB)= \sum_i f(\lambda_i)  \ketbra{u_i}{u_i}$. Letting $P$ be the Lagrange interpolation polynomial of degree $\leq n$ such that $P(\lambda_k)=f(\lambda_k)$ for $k=1,2,\ldots,n$, we observe that $f(AB)=P(AB)$ (where the coefficients of $P$ depend, in general, of $AB$). Therefore, by linearity of the trace, it suffices to prove the identity for $f(x)=x^k$.
Now since $Ag(A)=g(A)A$ and $Ah(A)=h(A)A$,
\begin{equation}
\begin{aligned}
\tr \bigl(g(A)(AB)^k h(A)\bigr) &=\tr \bigl(g(A)ABAB\cdots AB h(A)\bigr)\\
%&=\tr Ag(A) BAB\cdots AB h(A)\\
%&=\tr g(A) BAB\cdots AB h(A)A\\
&=\tr \bigl(g(A) BAB\cdots ABA h(A)\bigr)\\
&= \tr \bigl(g(A)(BA)^k h(A)\bigr)
\end{aligned}
\end{equation}
which ends the proof.
\end{proof}
In particular, if $f(x)=x\, g(x)$, we obtain
\begin{equation}\label{eq:identity2}
\begin{aligned}
\tr \bigl(A^{-1}f(BA)\bigr)&=\tr \bigl(A^{-1}f(AB)\bigr)\\
&=\tr\bigl( B\,g(AB)\bigr)=\tr\bigl( B\,g(BA)\bigr). 
\end{aligned}
 \end{equation}

\begin{lemma}\label{lem:opconvex} 
If $f$ is operator convex, then for any real numbers $x_i\geq 0$ and positive operators $\Lambda_i\geq 0$ such that $\sum_i \Lambda_i=I$, we have
\begin{equation}
f\bigl(\sum_i \Lambda_i x_i \bigr) \leq \sum_i \Lambda_i f(x_i).
\end{equation}
\end{lemma}
\begin{proof}
Apply Jensen's operator inequality~\cite{HansenPedersen2003}:
\begin{equation}
 f\bigl(\smash[b]{\sum_i} A_i^\dag X_i A_i \bigr) \leq \smash[b]{\sum_i} A_i^\dag f(X_i) A_i.
\end{equation}
where the $X_i\geq 0$ are positive operators and $\sum_i A_i^\dag A_i =I$,
to $X_i=x_i \cdot I$ and $A_i=\sqrt{\Lambda_i}$.
\end{proof}

%\section*{Acknowledgment}
%
%We are indebted to 

%%%%%%
%% To balance the columns at the last page of the paper use this
%% command:
%%
%\enlargethispage{-1.2cm} 
%%
%% If the balancing should occur in the middle of the references, use
%% the following trigger:
%%

\IEEEtriggeratref{10}

%%
%% which triggers a \newpage (i.e., new column) just before the given
%% reference number. Note that you need to adapt this if you modify
%% the paper.  The "triggered" command can be changed if desired:
%%
%\IEEEtriggercmd{\enlargethispage{-20cm}}
%%
%%%%%%

%%%%%%
%% References:
%% We recommend the usage of BibTeX:
%%
\bibliographystyle{IEEEtran}
\bibliography{paper}
%%
%%%%%%

\end{document}